\documentclass[letterpaper, 10 pt, conference]{ieeeconf} 
 
\IEEEoverridecommandlockouts                              
\usepackage[utf8]{inputenc}  
  
\usepackage{mathtools}
\mathtoolsset{showonlyrefs,showmanualtags}
\usepackage{amssymb,amsthm}
\usepackage{dsfont}
\usepackage{booktabs} 
\usepackage{bbold} 
\usepackage[yyyymmdd,hhmmss]{datetime}
\usepackage{bm,upgreek}
\usepackage{tablefootnote}
\usepackage{nth}

\usepackage{tikz}
\usetikzlibrary{shapes}
\usetikzlibrary{patterns}
\usepackage{multirow}
\usepackage{footnote}
\usepackage{color}
\def\BibTeX{{\rm B\kern-.05em{\sc i\kern-.025em b}\kern-.08em
    T\kern-.1667em\lower.7ex\hbox{E}\kern-.125emX}}

\usepackage{cite}
\newtheorem{assumption}{Assumption}
\newtheorem{theorem}{Theorem}
\newtheorem{definition}{Definition}
\newtheorem{remark}{Remark}
\newtheorem{lemma}{Lemma}

\newtheorem{proposition}{Proposition}
\newtheorem{problem}{Problem}

\DeclareMathOperator{\DIAG}{diag}

 \usepackage{tikz}
\usetikzlibrary{graphs,graphs.standard}

\tikzset{%
  every neuron/.style={
    circle,
    draw,
    minimum size=.7cm
  },
  neuron missing/.style={
    draw=none, 
    scale=3,
    text height=0.333cm,
    execute at begin node=\color{black}$\vdots$
  },
}

\definecolor{red}{RGB}{187,0,0}
\definecolor{blue}{RGB}{0, 0,180}
\definecolor{pink}{RGB}{203, 76, 178}

\title{\LARGE \bf
Receding Horizon Control  in Deep Structured Teams: A Provably Tractable Large-Scale Approach with Application to Swarm Robotics
}

\author{Jalal Arabneydi and Amir G. Aghdam
\thanks{
This work was supported in part by funding from the Innovation for  
Defence Excellence and Security (IDEaS) program from the Department of  
National Defence (DND). Any opinions and conclusions in this work are  those of the authors and do not reflect the views,  
positions, or policies of - and are not endorsed by - IDEaS, DND, or  
the Government of Canada.
}  
\thanks{Jalal Arabneydi and Amir G. Aghdam are with the  Department of Electrical and Computer Engineering, 
        Concordia University, 1455 de Maisonneuve Blvd. West, Montreal, QC, Canada, Postal Code: H3G 1M8.  Email: {\tt\small jalal.arabneydi@mail.mcgill.ca},        
        {\tt\small aghdam@ece.concordia.ca}}%
}

\begin{document}
\maketitle
\thispagestyle{empty}
\pagestyle{empty}

\vspace*{-5cm}{\footnotesize{Proceedings of IEEE  Conference on Decision and Control, 2021.}}
\vspace*{4.2cm}

\begin{abstract}
In this paper, a deep structured tracking problem is introduced for a large number of decision-makers. The problem is  formulated as a linear quadratic deep structured team, where  the decision-makers wish to  track a global target cooperatively while considering their local targets. For the unconstrained setup, the gauge transformation technique  is used to  decompose the resultant optimization problem in order to obtain a low-dimensional  optimal control strategy in terms of the local and global Riccati equations. For the constrained case, however, the feasible set is not necessarily decomposable  by the gauge transformation. To overcome this hurdle,  we propose a family of
local and global receding horizon control problems,  where a carefully constructed linear combination of their solutions provides a feasible solution for the original constrained problem. The salient property of the above solutions is that they  are tractable with respect to the number of decision-makers and can be implemented in a distributed  manner. In addition, the main results are generalized  to cases with  multiple sub-populations and multiple features, including  leader-follower setup,  cohesive cost function and soft structural constraint. Furthermore, a class of  cyber-physical attacks is proposed in terms of perturbed influence factors.  A numerical example  is presented to demonstrate the efficacy of the results.
\end{abstract}

%
\section{Introduction}

 Swarm tracking arises in many engineering applications such as robotics, smart grids and economics, where a group  of decision-makers wish to track a target collectively. To solve the swarm tracking problem, one common practice is to propose a strategy based on the  consensus algorithms, where the decision-makers are guaranteed to reach the target after a sufficiently large horizon~\cite{Olfati2007survey,tsitsiklis1986distributed, Jadbabaie2003,chung2018survey}. Alternatively,   one can define a  cost function consisting of the tracking cost (penalizing the distance between every agent and  the target) and the formation cost (penalizing the  relative distances between the agents).  Given a differentiable parametrized strategy, gradient decent methods can be utilized to search for a locally optimal solution~\cite{Zhang2012}.   On the other hand, it is difficult  to find a scalable solution  for large-scale swarms, in practice. This is  because  there is often a set of state and action constraints, leading to a non-trivial feasible set, such that  any naive  solution suffers from the curse of dimensionality with respect to the number of  decision-makers.

    To address the above shortcoming, we introduce \emph{deep structured tracking problem} wherein a large number of decision  makers  wish  to  track  a global target while taking into account their local targets.  The idea of deep structured tracking   stems from a newly emergent class of large-scale decentralized control systems called \emph{deep structured teams}~\cite{Jalal2019MFT,Jalal2019risk,Jalal2020CCTA,
    Jalal2020Automatica_On,Vida2020CDC,Masoud2020CDC,
    Jalal2021CDC_KF,Jalal2020Nash, Jalal2019TSNE}.  In deep structured teams/games,   decision-makers are coupled through a set of linear regressions of the states and actions of  the decision-makers, which is similar in spirit to the coupling of neurons in  feed-forward deep neural networks (DNN). For example,  it is shown in this paper that a feed-forward DNN with rectified linear unit activation function may be viewed as a special case of deep structured teams, where neurons are  agents with affine dynamics and affine constraints, and layers are time steps.  In general, a key step to obtain a low-dimensional solution for the  linear quadratic deep structured model is to  decompose the  optimization problem by a gauge transformation,  initially  proposed in~\cite{arabneydi2016new} and showcased in risk-sensitive model~\cite{Jalal2019risk}, decentralized estimation~\cite{Jalal2021CDC_KF}, reinforcement learning~\cite{Jalal2020CCTA,
    Vida2020CDC,Masoud2020CDC},~nonzero-sum game~\cite{Jalal2020Automatica_On},  minmax optimization~\cite{Jalal2019LCSS},  leader-follower tracking~\cite{JalalCCECE2018, JalalCDC2018}, and mean-field teams~\cite{JalalCDC2015, Jalal2017linear}.

To consider state and action constraints, we use receding horizon control in this article as a popular  industrial methodology, also known as  model predictive control,   rolling horizon planning, dynamic matrix control and dynamic linear programming~\cite{garcia1989model,qin2003survey,Alessio2009,Boyd2009}. In particular, we propose a family of two low-dimensional  receding horizon control problems,  where a carefully constructed linear combination of their solutions provides a feasible solution.
 In addition, we generalize our main results to include  multiple sub-populations, multiple features and cyber-physical attacks. 
In contrast to the  consensus-based algorithms, our approach  is a  control-based algorithm that is scalable with respect to the number of agents; see Subsection~\ref{sec:cohesive} for similarities and  differences between consensus and  (optimal) control algorithms. 

The remainder of the paper is organized as follows. In Section~\ref{sec:problem},  the problem is formulated and  in Section~\ref{sec:main},  the main results are obtained. In Sections~\ref{sec:multiple} and~\ref{sec:cyber}, the main results are extended  to   multiple sub-populations, multiple features and cyber-physical attacks.  A numerical example is presented in Section~\ref{sec:numerical} to verify  the obtained theoretical results.  In Section~\ref{sec:conclusions},  some conclusions  are drawn.

\section{Problem formulation}\label{sec:problem}
Throughout the paper, $\mathbb{R}$ and $\mathbb{N}$ refer to the sets of real and natural numbers, respectively. Given any $n \in \mathbb{N}$, $\mathbb{N}_n$ is the finite set $\{1,2,\ldots,n\}$. For  any vector $x$ and square matrix $Q$, $\|x\|_Q=x^\intercal Q x$. Short-hand notation $x_{a:b}$ denotes the set $\{x_a,\ldots,x_b\}$ for any $a \leq b \in \mathbb{N}$. All vector inequalities in this paper are element-wise, unless stated otherwise. Given two vectors $a$ and $b$  of the same size, $\min(a,b)$ and $\max(a,b)$ refer, respectively, to a vector whose elements are the minimum and maximum of the elements of  $a$ and~$b$.

Consider a swarm of  $n \in \mathbb{N}$ decision-makers (agents). Let $x^i_t \in \mathbb{R}^{d_x}$ and $u^i_t \in \mathbb{R}^{d_u}$ denote the state and action of agent $i \in \mathbb{N}_n$ at time $t \in \mathbb{N}$. Let $\alpha_i \in \mathbb{R}$ denote the \emph{influence factor} of  agent $i$ at the focal point  of the swarm such that
\begin{equation}\label{eq:deep_state_alpha}
\bar x^\alpha_t:= \frac{1}{n}\sum_{i=1}^n \alpha_{i} x^i_t,
\end{equation}
where $\bar x^\alpha_t$ is called the \emph{center of swarm} at time $t$.  Similarly, define the following linear regression in the action space:
\begin{equation}\label{eq:deep_action_alpha}
\bar u^\alpha_t:= \frac{1}{n}\sum_{i=1}^n \alpha_{i} u^i_t.
\end{equation}

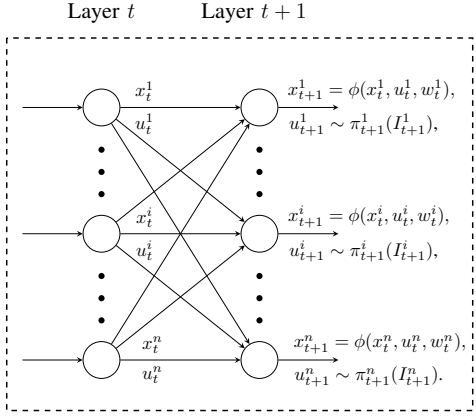
\begin{figure}[t!]
\hspace*{0.2cm}
\scalebox{0.7}{
\begin{tikzpicture}[x=1.5cm, y=1.2cm, >=stealth]
  \draw[font=\large](0,3cm) node {Layer $t$} ;
\draw[font=\large](2.9cm,3cm) node {Layer $t+1$} ;
\draw[font=\large](6.5cm,-5cm) node {} ;
\draw[black, thick, dashed] (-1.2,2.1) rectangle (4.7,-3.8);

\foreach \m/\l [count=\y] in {1,missing,2,missing,3}
  \node [every neuron/.try, neuron \m/.try] (input-\m) at (0,2-\y) {};

\foreach \m/\l [count=\y] in {1,missing,2,missing,3}
  \node [every neuron/.try, neuron \m/.try ] (output-\m) at (2,2-\y) {};

  \draw [<-] (input-1) -- ++(-1,0)
    node [above,midway] {\hspace{3.5cm}$x^1_t$}
        node [below,midway] {\hspace{3.5cm}$u^1_t$};
     \draw [<-] (input-2) -- ++(-1,0)
    node [above,midway] {\hspace{3.5cm}$x^i_t$}
        node [below,midway] {\hspace{3.5cm}$u^i_t$};
        
          \draw [<-] (input-3) -- ++(-1,0)
    node [above,midway] {\hspace{3.8cm}$x^{n}_t$}
        node [below,midway] {\hspace{3.8cm}$u^{n}_t$};

  \draw [->]  (output-1) -- ++(1,0)  node[above, midway]  
   {\hspace*{2.4cm}$x^1_{t+1}= \phi(x^1_t,u^1_t,w^1_t),$}
     node [below,midway] {\hspace{2.1cm}$u^1_{t+1} \sim \pi^1_{t+1}( I^1_{t+1})$,};   
      \draw [->]  (output-2) -- ++(1,0)  node[above, midway]     {\hspace*{2.3cm}$x^i_{t+1}= \phi(x^i_t,u^i_t,w^i_t),$}
     node [below,midway] {\hspace{2.1cm}$u^i_{t+1} \sim \pi^i_{t+1}( I^i_{t+1})$,};  
    
      \draw [->]  (output-3) -- ++(1,0)  node[above, midway]  { \hspace*{2.6cm} $x^{n}_{t+1}= \phi(x^{n}_t,u^{n}_t,w^{n}_t),$}
     node [below,midway] {\hspace{2.3cm}$u^{n}_{t+1} \sim \pi^{n}_{t+1}( I^{n}_{t+1})$.};     
     
    \foreach \i in {1,...,3}
  \foreach \j in {1,...,3}
    \draw [->] (input-\i) -- (output-\j);    
\end{tikzpicture}}
\caption{The interaction (coupling) between agents in deep structured teams is similar in spirit to that of neurons in a
deep feed-forward neural network.  In this paper, the dynamics $\phi$ is a deterministic affine function, $I^i_t$ is the information set of agent $i$, and strategy $\pi$ is computed by  a set of local and global Riccati equations and quadratic programmings for the unconstrained and  constrained cases, respectively.  }\label{fig:DNN}
\end{figure}

From the terminology of deep structured teams, aggregate variables $\bar x^\alpha_t$ and $\bar u^\alpha_t$ are also   referred to as  \emph{deep state} and \emph{deep action}, respectively. The reason for such naming is that $\bar x^\alpha_t$ and $\bar u^\alpha_t$, $t\geq 1$, may be viewed as a mapping from  the initial states (the input of control system)  to a real-valued vector, where the mapping is constructed  by~$t$ sequential layers of some  parallel operations, resembling  a feed-forward deep neural network (DNN); see Figures~\ref{fig:DNN} and~\ref{fig:DNN2}. Notice that $\bar x^\alpha_t$ and $\bar u^\alpha_t$ are normalized with respect to~$n$  because we are interested in  applications with large $n$.  Let  $\bar \alpha$ denote  the  average of influence factors, i.e.
$\bar{\alpha}:=\frac{1}{n}\sum_{i=1}^n \alpha_i$.

\begin{definition}[Center of mass]\label{def:center}
The center of swarm is called the center of mass if $\bar \alpha=1$. In such a case, there exists a set of  scalars $\{m_i \in \mathbb{R}, \forall i \in \mathbb{N}_n\}$ such that
\begin{equation}
\bar x^\alpha_t:= \frac{1}{M}\sum_{i=1}^n m_{i} x^i_t,
\end{equation}
where  $m_i:=\frac{1}{n}\alpha_i M$ and $M:=\frac{1}{n} \sum_{i=1}^n m_i$.  An important special case of the center of mass is where  $\{\alpha_i=\frac{nm_i}{M}\geq~0, \forall i \in \mathbb{N}_n\}$ is a convex combination of scalars.
\end{definition}
 At  time $t \in \mathbb{N}$, the state of agent $i \in \mathbb{N}_n$ evolves as:
\begin{equation}\label{eq:dynamics}
x^i_{t+1}=A_t x^i_t +B_t u^i_t,
\end{equation}
where $A_t$ and $B_t$ are matrices with appropriate dimensions.

\subsection{Cost function}

Let $r^i_t \in \mathbb{R}^{d_x}$ denote  the local reference of agent $i$ at time $t \in \mathbb{N}_T$ indicating the center of its safe zone  and $s_t \in \mathbb{R}^{d_x}$ denote the the global reference  of the swarm determining the desired trajectory  of the center of agents. To this end, we define a cost function with  a \emph{common} penalty function penalizing the  mismatch between the center of swarm $\bar x^\alpha_t$ and the global reference $s_t$. 
 More precisely, for any $i \in \mathbb{N}_n$ and $t \in \mathbb{N}_T$, the cost of agent $i \in \mathbb{N}_n$ is defined as:
\begin{equation}\label{eq:per_step_cost}
c^i_t=\gamma_i(\|x^i_t - r^i_t\|_{Q_t}+ \|u^i_t\|_{R_t})+ \|\bar{x}^\alpha_t - s_t\|_{\bar Q_t} +  \|\bar u^\alpha_t\|_{\bar R_t}, 
\end{equation}
where $\gamma_i>0$ denotes the importance of the local cost of agent $i$ and matrices  $Q_t$, $R_t$,  $\bar Q_t$ and $\bar R_t$ are symmetric  with  appropriate dimensions.  The first term  forces  agent $i$ to be close to its safe zone whose center is given by  $r^i_t$ and  the second one considers  energy consumption of agent $i$. The third term incentivizes  the center of  the swarm to track  the global target whereas  the forth term (which can be set to zero, i.e. $\bar R_t=\mathbf 0$) smooths the trajectory of the center of  the swarm by preferring  small  values for the deep action $\bar u^\alpha_t$.

 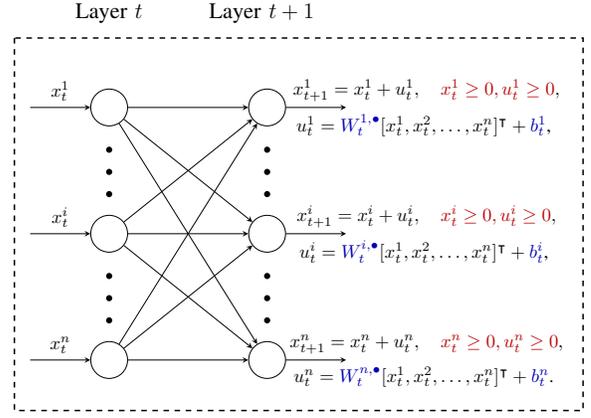
\begin{figure}[t!]
 \scalebox{.7}{
 \hspace{.6cm}
\begin{tikzpicture}[x=1.5cm, y=1.2cm, >=stealth]
  \draw[font=\large](0,3cm) node {Layer $t$} ;
\draw[font=\large](2.9cm,3cm) node {Layer $t+1$} ;
\draw[font=\large](6.5cm,-5cm) node {} ;
\draw[black, thick, dashed] (-1.2,2.1) rectangle (6,-3.8);

\foreach \m/\l [count=\y] in {1,missing,2,missing,3}
  \node [every neuron/.try, neuron \m/.try] (input-\m) at (0,2-\y) {};

\foreach \m/\l [count=\y] in {1,missing,2,missing,3}
  \node [every neuron/.try, neuron \m/.try ] (output-\m) at (2,2-\y) {};

\foreach \l [count=\i] in {1,i,n}
  \draw [<-] (input-\i) -- ++(-1,0)
    node [above, midway] {$x^\l_t$};

  \draw [->]  (output-1) -- ++(1,0)  node[above, midway]   {\hspace*{4.2cm}$x^1_{t+1}=x^1_t+u^1_t, \quad \textcolor{red}{x^1_t \geq 0, u^1_t \geq 0},$} node [below,midway] {\hspace{4.1cm}$u^1_t=\textcolor{blue}{W^{1, \bullet}_t}  [ x^1_t, x^2_t, \ldots,x^n_t]^\intercal + \textcolor{blue}{b^1_t}$,};   
    
      \draw [->]  (output-2) -- ++(1,0)  node[above, midway]   {\hspace*{4.2cm}$x^i_{t+1}=x^i_t+u^i_t, \quad \textcolor{red}{x^i_t \geq 0, u^i_t \geq 0},$} node [below,midway] {\hspace{4.1cm}$u^i_t=\textcolor{blue}{W^{i, \bullet}_t}  [ x^1_t, x^2_t, \ldots,x^n_t]^\intercal + \textcolor{blue}{b^i_t}$,};
    
      \draw [->]  (output-3) -- ++(1,0)  node[above, midway]   {\hspace*{4.2cm}$x^n_{t+1}=x^n_t+u^n_t, \quad \textcolor{red}{x^n_t \geq 0, u^n_t \geq 0},$} node [below,midway] {\hspace{4.1cm}$u^n_t=\textcolor{blue}{W^{n, \bullet}_t}  [ x^1_t, x^2_t, \ldots,x^n_t]^\intercal + \textcolor{blue}{b^n_t}$.};
    
    \foreach \i in {1,...,3}
  \foreach \j in {1,...,3}
    \draw [->] (input-\i) -- (output-\j);
\end{tikzpicture}}
\caption{A feed-forward DNN with Rectified Linear Unit (ReLU) activation function may be viewed as a special case of deep structured teams, where neurons are agents, and layers are time steps. In particular, the dynamics of agents is a single integrator, which is an affine function, with affine constraints, wherein  $W_t$ and $b_t$ represent the weight matrix and bias vector, respectively.  An alternative   formulation of  DNN with ReLU function is  where $A_t=0$ and $x^i_{t+1}=u^i_t, u^i_t \geq 0, \forall i \in~\mathbb{N}_n$.}\label{fig:DNN2}
\vspace{-.4cm}
\end{figure}
 \subsection{Problem statement}
The agents are interested to collaborate to minimize a common cost function defined as
\begin{equation}\label{eq:total_cost}
J_n:= \frac{1}{n} \sum_{t=1}^T\sum_{i=1}^n  c^i_t(x^1_{t},\ldots,x^n_{t},u^1_{t},\ldots,u^n_{t}).
\end{equation}

 \begin{remark}
\emph{Note  that  our main results hold for any setup in which the per-step cost in~\eqref{eq:total_cost} can be represented as a summation of local cost functions  (in terms of local states and local actions) and global cost functions (in terms of deep states and deep actions). Below, we present two such cases.
\begin{itemize}
\item Any weighted cross-terms in $c^i_t$ can be formulated as:
\begin{equation}
\frac{1}{n}\sum_{i=1}^n  \alpha_i (x^i_t)^\intercal Q_t \bar x^\alpha_t=\|\bar x^\alpha_t\|_{Q_t}.
 \end{equation}
\item Any weighted  tracking cost  can be expressed as:
\begin{multline}
\frac{1}{n}\sum_{i=1}^n  \alpha_i \|x^i_t -F_t \bar x^\alpha_t\|_{Q_t}=\frac{1}{n} \sum_{i=1}^n \alpha_i\| x^i_t\|_{Q_t}+ \|\bar x^\alpha_t  \|_{\bar Q_t},
 \end{multline}
where $\bar Q_t:= (I-F_t)^\intercal Q (I- F_t) -Q_t$.
\end{itemize}
}
\end{remark}

\begin{problem}[Optimal control]
Find a scalable  optimal strategy such that the team cost in~\eqref{eq:total_cost} is minimized, i.e.,
\begin{align}
&\hspace{-1cm}J^\ast_n:=\min_{u^1_{1:T},\ldots,u^n_{1:T}}  \frac{1}{n}\sum_{t=1}^T \sum_{i=1}^n c^i_t(x^1_{t},\ldots,x^n_{t},u^1_{t},\ldots,u^n_{t}),\\
\text{\emph{subject to:}}& \quad x^i_{t+1}=A_t x^i_t+B_t u^i_t, \quad \forall i \in \mathbb{N}_n, \forall t \in \mathbb{N}_T.
\end{align}
\end{problem}

\begin{problem}[Receding horizon control (RHC)]
Develop a scalable  RHC  for the following constrained optimization:
\begin{align}
&\hspace{-1cm} \min_{u^1_{1:T},\ldots,u^n_{1:T}}   \frac{1}{n}\sum_{t=1}^T \sum_{i=1}^n  c^i_t(x^1_{t},\ldots,x^n_{t},u^1_{t},\ldots,u^n_{t}),\\
\text{\emph{subject to:}}& \quad x^i_{t+1}=A_t x^i_t+B_t u^i_t, \quad \forall i \in \mathbb{N}_n,  \forall t \in \mathbb{N}_T,\\
&\quad a \leq  x^i_t \leq  b, \quad a, b \in \mathbb{R}^{d_x} ,\\
&\quad c \leq u^i_t \leq d, \quad c, d \in \mathbb{R}^{d_u},\\
& \quad \bar a \leq \bar x^\alpha_t \leq  \bar b, \quad \bar a, \bar b \in \mathbb{R}^{d_x}, \\
& \quad \bar c \leq \bar u^\alpha_t \leq  \bar d, \quad \bar c, \bar d \in \mathbb{R}^{d_x}.
\end{align}
\end{problem}

For the special case of  non-negative influence factors (e.g. convex combination),   the effective lower and upper bounds  imposed on the  deep state and deep  action in Problem~2 are $\max(\bar \alpha a, \bar a)$, $\max(\bar \alpha c, \bar c)$, $\min(\bar \alpha b, \bar b)$ and $\min(\bar \alpha d, \bar d)$.

\subsection{Main challenges and contributions}
The first challenge is the curse of dimensionality with respect to the number of agents, where the augmented matrices are fully dense.  To overcome this challenge, we use a gauge transformation (i.e.,  a change of coordinates) to decompose the optimization problem  in order to  obtain  a low-dimensional solution in terms of two scale-free Riccati equations. We show that  the centralized solution can be implemented in a distributed  manner wherein every agent needs  access to  only  the deep state  (rather than the entire joint state). The second challenge is  that the  feasible set of the constrained optimization problem (Problem 2) is not fully decomposable by the gauge transformation, which means that the solution of Problem~1 is not directly applicable in this case. To this end, we propose two scale-free  RHCs  under mild conditions   for every agent. We show that a carefully constructed   linear combination of the solutions of the  proposed RHCs provides a feasible solution for  Problem~2.
\section{Main results for Problems 1 and 2}\label{sec:main}
In this section, we present the main results for Problems~1 and 2. Prior to delving into theoretical results, we define two types of tracking as follows.

\begin{definition}[Strong and weak swarm tracking]\label{def:strong}
When the center of swarm reduces to the center of mass,  the swarm tracking is called  strong; otherwise,   it is called  weak.
\end{definition}

\begin{proposition}
Suppose that the tracking is weak and the center of swarm is not at the origin. Then, there is at  least one agent that does not converge to  the  center of swarm.
  \end{proposition}
\begin{proof}
The proof follows from contradiction. Suppose all agents converge to the center of swarm at some time $t \in \mathbb{N}$, i.e.  $x^i_t=\bar x^\alpha_t$, $\forall i \in \mathbb{N}_n$. From~\eqref{eq:deep_state_alpha},  one has $ (\frac{1}{n}\sum_{i=1}^n \alpha_i) \bar x^\alpha_t= \bar x^\alpha_t$, which holds if and only if   $\bar x^\alpha_t=~0$.
\end{proof}

 In general,  weak  tracking arises in various situations wherein the center of swarm is not properly balanced.  This unbalanced property  may be caused by an external force (e.g., cyber-physical attack) or by the designer (e.g., when agents wish to  monitor a target without getting close to it).

\subsection{Gauge transformation and Riccati equation}
The first step to solve a linear quadratic deep structured team is to use a gauge transformation, initially introduced in~\cite{arabneydi2016new} and showcased in~\cite{Jalal2020CCTA}, to define auxiliary variables as the deviation of the local variables from  deep (weighted) variables.  We  use the following gauge transformation:
\begin{equation}\label{eq:gauge}
\Delta x^i_t:= x^i_t - \frac{\alpha_i}{\gamma_i} \bar x^\alpha_t, \quad \Delta u^i_t:= u^i_t - \frac{\alpha_i}{\gamma_i} \bar u^\alpha_t, \quad \Delta r^i_t:= r^i_t - \frac{\alpha_i}{\gamma_i} \bar r^\alpha_t,
\end{equation}
where
 $\bar r^\alpha_t:=\frac{1}{n}\sum_{i=1}^n \alpha_i r^i_t$.
From~\eqref{eq:dynamics}, one has
\begin{equation}\label{eq:global_dynamics}
\Delta x^i_{t+1}=A_t \Delta x^i_t+ B_t \Delta u^i_t,
\end{equation}
and
\begin{equation}\label{eq:local_dynamics}
\bar x^\alpha_{t+1}=A_t \bar x^\alpha_t + B_t \bar u^\alpha_t.
\end{equation}
Let $\mu:= \frac{1}{n} \sum_{i=1}^n \frac{\alpha_i^2}{\gamma_i}$.
\begin{lemma}\label{lemma:DST_cost}
 The per-step cost function in equation~\eqref{eq:total_cost} at any  time $t \in \mathbb{N}$ can be written as:
\begin{align}
&\Big(\frac{1}{n} \sum_{i=1}^n \gamma_i( \|\Delta x^i_t - \Delta r^i_t\|_{Q_t} + \| \Delta u^i_t\|_{R_t})\Big) + \| \bar x^\alpha_t -s_t \|_{\bar Q_t}\\
&+
(2 -\mu)\|\bar x^\alpha_t - \bar r^\alpha_t \|_{Q_t} +  \|\bar u^\alpha_t\|_{\bar R_t}+ (2 -\mu) \|\bar u^\alpha_t\|_{R_t}.
\end{align}
\end{lemma}
\emph{Proof.}
The proof follows directly from~\eqref{eq:per_step_cost}, the gauge transformation~\eqref{eq:gauge} and the fact that
\begin{equation}
\frac{1}{n} \sum_{i=1}^n \alpha_i (\Delta x^i_t -\Delta r^i_t)^\intercal Q_t (\bar x^\alpha_t - \bar r^\alpha_t)=(1 -\mu) \|\bar x^\alpha_t - \bar r^\alpha_t\|_{Q_t},
\end{equation}
and
\begin{equation}
\frac{1}{n} \sum_{i=1}^n \alpha_i (\Delta u^i_t)^\intercal R_t (\bar u^\alpha_t)=(1 -\mu) \|\bar u^\alpha_t \|_{R_t}.  \hspace{1.9cm} \blacksquare
\end{equation}
Denote $\mathbf Q_t:= (2 -\mu) Q_t +  \bar Q_t$ and $\mathbf R_t:=(2 -\mu) R_t+  \bar R_t$.
\begin{assumption}\label{ass:convexity_DST}
At any time $t \in \mathbb{N}$,  $Q_t$ and $\mathbf Q_t$ are positive semi-definite and $ R_t$ and $\mathbf R_t$ are positive definite.
\end{assumption}
Define  local and global Riccati equations as follows:
\begin{equation}\label{eq:Riccati}
\begin{cases}
P_{t}= Q_t+ A_t^\intercal P_{t+1} A_t - A_t^\intercal P_{t+1} B_t (B_t^\intercal P_{t+1} B_t+ R_t)^{-1}\\
\qquad \times B_t^\intercal P_{t+1} A_t,  \forall t \in \mathbb{N}_{T-1},\\
\mathbf P_{t}= \mathbf Q_t+ A_t^\intercal \mathbf P_{t+1} A_t - A_t^\intercal \mathbf P_{t+1} B_t (B_t^\intercal \mathbf P_{t+1} B_t+ \mathbf R_t)^{-1}\\
\qquad \times B_t^\intercal \mathbf P_{t+1} A_t,  \forall t \in \mathbb{N}_{T-1},
\end{cases}
\end{equation} 
with $P_T=Q_T$ and $\mathbf P_T=\mathbf Q_T$.

\begin{theorem}\label{thm:1}
Let Assumption~\ref{ass:convexity_DST} hold. The optimal strategy  of agent $i \in \mathbb{N}_n$ at any time $t \in \mathbb{N}_{T-1}$ is given by:
\begin{equation}
u^{\ast,i}_t= \theta^\ast_t x^i_t + \frac{\alpha_i}{\gamma_i} (\bar \theta^\ast_t - \theta^\ast_t) \bar x^\alpha_t+ L_{t}v^i_{t+1}+ \frac{\alpha^i}{\gamma_i} \bar L_{t}\bar v^\alpha_{t+1},
\end{equation}
where  gain matrices $\{\theta^\ast_t, \bar \theta^\ast_t, L_t,\bar L_{t}\}$ and correction signals $\{ \{v^i_t \}_{i=1}^n, \bar v_t\}$ are obtained from  the solution of the local and global Riccati equations~\eqref{eq:Riccati}  where for  any $t \in \mathbb{N}_{T-1}$:
\begin{equation}
\begin{cases}
 \theta^\ast_t:=-(B_t^\intercal P_{t+1} B_t + R_t)^{-1} B_t^\intercal  P_{t+1} A_t,\\
 L_t:= (B_t^\intercal  P_{t+1} B_t +  R_t)^{-1} B_t^\intercal,\\
\bar \theta^\ast_t:=-(B_t^\intercal \mathbf P_{t+1} B_t + \mathbf R_t)^{-1} B_t^\intercal \mathbf P_{t+1} A_t,\\
\bar L_t:= (B_t^\intercal \mathbf P_{t+1} B_t + \mathbf R_t)^{-1} B_t^\intercal,\\
\end{cases}
\end{equation}
and 
\begin{equation}
\begin{cases}
v^i_{t}:=(A_t + B_t \theta^\ast_{t})^\intercal v^i_{t+1} + Q_t \Delta r^i_t,  \quad i \in \mathbb{N}_n, \\
v^i_{T}:=Q_T \Delta r^i_{T}, \quad  i \in \mathbb{N}_n, \\
\bar v_{t}:=(A_t + B_t \bar \theta^\ast_{t})^\intercal \bar v_{t+1} + (2- \mu) Q_t \bar r^\alpha_t+  \bar Q_t s_t, \\
\bar v_{T}:=(2- \mu) Q_T \bar r^\alpha_T + \bar Q_T s_T.
\end{cases}
\end{equation}
\end{theorem}
\emph{Proof.}
The proof follows from equations~\eqref{eq:global_dynamics} and~\eqref{eq:local_dynamics} and Lemma~\ref{lemma:DST_cost}, where the optimization in~Problem~1 can be decomposed to $n+1$ smaller optimizations. More precisely,  there is a local  linear quadratic regulator (LQR) for every $i \in \mathbb{N}_n$ with state and action $\{\Delta x^i_t, \Delta u^i_t \}$ and tracking signal $\{\Delta r^i_t\}$. Since  $\gamma_i >0$, it does not affect the optimization problem. Therefore, one has for any $i \in \mathbb{N}_n$:
$\Delta u^{\ast,i}_t= \theta^\ast_t \Delta x^i_t+ L_t v^i_{t+1}
$. There is also one LQR  with  state and action $\{\bar  x^\alpha_t, \bar u^\alpha_t \}$ and tracking signals $\bar r^\alpha_t$ and $s_t$, where
$\bar  u^{\ast,\alpha}_t= \bar \theta^\ast_t \bar  x_t+ \bar L_t \bar v_{t+1}$. 
From  gauge transformation~\eqref{eq:gauge}, it results that:
\begin{equation}
\hspace{1.8cm} u^{\ast,i}_t=\Delta u^{\ast,i}_t + \frac{\alpha_i}{\gamma_i} \bar u^{\ast,\alpha}_t, i \in \mathbb{N}_n. \hspace{2cm} \blacksquare
\end{equation}

\subsection{Receding horizon control}
A naive way to solve the  centralized RHC in Problem 2 leads to a large-scale optimization problem that is intractable with respect to the number of agents. In addition,  the centralized RHC  does not necessarily decompose into scalable problems after the gauge transformation. This  is in contrast to the unconstrained model wherein the centralized solution  coincides with  two scalable optimal control problems. To overcome this hurdle, we propose two scalable RHC problems whose feasible sets are a subset of the feasible set of the  centralized RHC problem.  In particular,   to distinguish  between the  state and action  of the proposed RHC problems and those of the original Problem~2, we use notations $y$ and $v$ instead of $x$ and $u$, respectively. Define one local and one global RHC  problem as follows.

\begin{problem}[Local RHC]
For any agent $i \in \mathbb{N}_n$ and horizon $H \in \mathbb{N}$, find a solution for the following minimization:
\begin{align}
\hspace{-2cm}&\min_{\Delta v^i_{t:t+H}} \sum_{\tau=t}^{t+H} \|\Delta y^i_\tau - \Delta r^i_\tau \|_{Q_\tau}+ \|\Delta v^i_\tau \|_{R_\tau},\\
\text{\emph{s.t.}}\quad &  \Delta y^i_{\tau+1}=A_\tau \Delta y^i_\tau+B_\tau \Delta v^i_\tau,  \tau \in \{t,\ldots,t+H-1\},\\
&\quad \tilde a_i \leq  \Delta y^i_\tau \leq  \tilde b_i, \quad \tilde a_i, \tilde b_i \in \mathbb{R}^{d_x} ,\\
&\quad \tilde c_i \leq \Delta v^i_\tau \leq \tilde d_i, \quad \tilde c_i, \tilde d_i \in \mathbb{R}^{d_u}.
\end{align}
\end{problem}
\begin{problem}[Global RHC]
 Given any prediction horizon $H \in \mathbb{N}$, find a solution for the following minimization:
\begin{align}\label{eq:RHC2}
& \hspace{-.9cm}\min_{\bar v^\alpha_{t:t+H}} \sum_{\tau=t}^{t+H} \|\bar y^\alpha_\tau - \bar r^\alpha_\tau \|_{ (2- \mu) Q_\tau} + \|\bar y^\alpha_\tau - s_\tau \|_{  \bar Q_\tau}+ \|\bar v^\alpha_\tau \|_{\mathbf R_\tau}, \nonumber \\
\text{\emph{s.t.}}\quad &  \bar y^\alpha_{\tau+1}=A_\tau \bar y^\alpha_\tau+B_\tau \bar v^\alpha_\tau,  \tau \in \{t,\ldots,t+H-1\},\nonumber \\
&\quad  \underline a \leq  \bar  y^\alpha_\tau \leq  \underline b, \quad  \underline a, \underline b \in \mathbb{R}^{d_x},\\
&\quad \underline c \leq \bar v^\alpha_\tau \leq \underline d, \quad \underline c, \underline d \in \mathbb{R}^{d_u}.
\end{align}
\end{problem}

\begin{remark}
\emph{At any time $t \in \mathbb{N}_T$,  one can solve the above  open-loop control problems  by quadratic programming.  Notice that the feasible set of the proposed  RHC Problems 3 and 4 is not necessary equal to that of  the  RHC Problem~2.}
\end{remark}

 Now, we introduce a family of  bounds for Problems 3 and 4 such that their solution is valid for Problem 2. 
\begin{assumption}\label{ass:param1}
Let  $\alpha_i  \in (0,1],\alpha_i  \leq  \gamma_i$,  $\forall i \in \mathbb{N}_n$.  Let also $b, \bar b > \mathbf 0_{d_x}$,  $d, \bar d> \mathbf 0_{d_u}$, $a, \bar a < \mathbf 0_{d_u}$ and $d, \bar d < \mathbf 0_{d_u}$.
\end{assumption}

\begin{theorem}\label{thm:2}
Let Assumptions~\ref{ass:convexity_DST} and~\ref{ass:param1} hold.  For  any $\lambda \in (0,1)$ and $i \in \mathbb{N}_n$,   suppose  that the boundaries of the local and global RHC Problems 3 and 4 are given by:
\begin{align}\label{eq:param1}
\tilde a_i &:= \frac{\lambda}{1-\lambda} \underline a, \quad \tilde b_i:=  \frac{\lambda}{1-\lambda}  \underline{b},\\
\tilde c_i &:= \frac{\lambda}{1-\lambda}  \underline c, \quad \tilde d_i:= \frac{\lambda}{1-\lambda}  \underline d,\\
\underline a &:=(1-\lambda) \max(\bar \alpha a, \bar a), \quad\underline b:= (1-\lambda) \min(\bar \alpha b, \bar b),\\
\underline c &:=(1-\lambda)\max(\bar \alpha c, \bar c), \quad\underline d:= (1-\lambda) \min(\bar \alpha d, \bar d).
\end{align}
Then, at any time $t \in \mathbb{N}_T$, the following linear combination: 
\begin{equation}
u^i_t= \Delta v^i_t + \frac{\alpha_i}{\gamma_i} \bar v^\alpha_t, \quad i \in \mathbb{N}_n,
\end{equation}
is a feasible solution for  Problem~2.
\end{theorem}
\begin{proof}
In the first step, we show that the above limits construct a non-empty set. To avoid repetition, we only prove  the results for those constraints imposed on state spaces because similar arguments hold for action spaces. Since $1-\lambda>0$,  $0< \bar \alpha \leq 1$,  $\max(\bar \alpha a, \bar a) < \boldsymbol{0}_{d_x}$ and $\min(\bar \alpha b, \bar b)> \boldsymbol{0}_{d_x}$, one can conclude that
 for every $i \in \mathbb{N}_n$,
\begin{equation}
\underline{a} < \boldsymbol{0}_{d_x} < \underline{b} \quad \text{and} \quad  \tilde a_i =\frac{\lambda}{1-\lambda} \underline a <  \mathbf 0_{d_x}<
\frac{\lambda}{1-\lambda} \underline b= \tilde b_i.
\end{equation}
In the second step, we show that the above limits present a feasible set for Problem~2. By definition, for any $i \in \mathbb{N}_n$ at time $t \in \mathbb{N}$:
\begin{equation}\label{eq:proof_2}
  \tilde a_i \leq \Delta y^i_t \leq \tilde b_i  \quad \text{and}\quad  \underline  a \leq  \bar y^\alpha_t \leq  \underline  b.
\end{equation}
Therefore, one arrives at:
\begin{equation}\label{eq:proof_3}
\tilde a_i + \frac{\alpha_i}{\gamma_i} \underline a \leq y^i_t:=\Delta y^i_t + \frac{\alpha_i}{\gamma_i} \bar y^\alpha_t   \leq \tilde b_i+ \frac{\alpha_i}{\gamma_i} \underline b,
\end{equation}
where the left-hand side of~\eqref{eq:proof_3} is lower-bounded by
\begin{equation}
(\frac{\lambda}{1-\lambda})\underline a +\frac{\alpha_i}{\gamma_i}\underline{a} \geq (\frac{\lambda}{1-\lambda}+1)\underline a=\max(\bar \alpha a, \bar a) \geq \bar \alpha a \geq a
\end{equation}
and the right-hand side of~\eqref{eq:proof_3}  is upper-bounded by
\begin{equation}
(\frac{\lambda}{1-\lambda})\underline b +\frac{\alpha_i}{\gamma_i}\underline{b}  \leq(\frac{\lambda}{1-\lambda}+1)\underline{b} = \min(\bar \alpha b, \bar b)\leq \bar \alpha b \leq b.
\end{equation}
As a result, one can conclude that the lower and upper bounds on local states in Problem~2 are satisfied, i.e., 
$a \leq y^i_t   \leq b$.
In addition, it is straightforward to show that   the lower and upper bounds on the deep state in Problem~2 is satisfied, where
\begin{align}
&\bar a \leq \max(\bar \alpha a, \bar a) \leq \frac{\lambda}{1-\lambda} \underline{a}+\underline{a}  \leq  \frac{1}{n} \sum_{i=1}^n \alpha_i y^i_t =\\
&\frac{1}{n} \sum_{i=1}^n \alpha_i \Delta y^i_t+ \frac{1}{n} \sum_{i=1}^n \frac{\alpha_i^2}{\gamma_i} \bar y^\alpha_t
 \leq  \frac{\lambda}{1-\lambda} \underline{b}+\underline{b} \leq \min(\bar \alpha b, \bar b) \leq \bar b.
\end{align}
Thus,  solution of  Problems 3-4 is  feasible  for Problem~2.
\end{proof}
\begin{remark}
\emph{Consider a special case when influence factors are a convex combination, i.e.  $\alpha_i \in 
(0,1], \forall i \in \mathbb{N}_n,$ and  $\bar \alpha=1$ such that $\alpha_i \leq \gamma_i$, $\lambda=\frac{1}{2}$, $a=\bar a=-b=-\bar b$ and $c=\bar c=-d=-\bar d$. From Theorem~\ref{thm:2}, one can show that the following bounds provide a feasible solution: $\tilde{a_i}=-\tilde b_i= \underline{a}=-\underline{b}=\frac{1}{2} a$ and $\tilde{c_i}=-\tilde d_i= \underline{c}=-\underline{d}=\frac{1}{2} c, \forall i \in \mathbb{N}_n$.}
\end{remark}
In contrast to Theorem~\ref{thm:2} that only holds for positive factors, we present a new theorem  with more conservative bounds  including both  negative and positive factors. Define $m_x:=\min(b,\bar b)$ if $\min(b,\bar b) + \max(a, \bar a) < \mathbf 0$ and $m_x:=-\max(a, \bar a)$ if $\min(b,\bar b) + \max(a, \bar a) > \mathbf 0$. Similarly,  define $m_u:=\min(d,\bar d)$ if $\min(d,\bar d) + \max(c, \bar c) <\mathbf  0$ and $m_u:=-\max(c, \bar c)$ if $\min(d,\bar d) + \max(c, \bar c) > \mathbf 0$.
\begin{assumption}\label{ass:param2}
Let  $\alpha_i \in [-1,1] $ and $\alpha_i \leq \gamma_i $,  $\forall i \in \mathbb{N}_n$.  Let also $b, \bar b > \mathbf 0_{d_x}$,  $d, \bar d> \mathbf 0_{d_u}$, $a, \bar a < \mathbf 0_{d_u}$ and $d, \bar d < \mathbf 0_{d_u}$.
\end{assumption}
\begin{theorem}\label{thm:3}
Let Assumptions~\ref{ass:convexity_DST} and~\ref{ass:param2} hold. 
 For  any $\lambda \in (0,1)$ and $i \in \mathbb{N}_n$,  suppose that the boundaries of the local and global RHC Problems 3 and 4 are given by:
\begin{align}
&-\tilde a_i :=\tilde b_i:= \lambda m_x, \quad 
-\tilde c_i :=\tilde d_i:= \lambda  m_u,\\
&-\underline a:=\underline b :=(1-\lambda) m_x, \quad
-\underline c :=\underline d:=(1-\lambda) m_u.
\end{align}
Then, at any time $t \in \mathbb{N}_T$, the following linear combination:  
\begin{equation}
u^i_t= \Delta v^i_t + \frac{\alpha_i}{\gamma_i} \bar v^\alpha_t, \quad i \in \mathbb{N}_n,
\end{equation}
is a feasible solution for  Problem~2.
\end{theorem}
\begin{proof}
The proof proceeds along the same lines as the proof of Theorem~\ref{thm:2}. In the first step, we show that the above limits construct a non-empty set, i.e.,  from the definition of $m_x$, 
\begin{equation}
a\leq -m_x, \bar a \leq -m_x, m_x \leq b, m_x < \bar b,
\end{equation}
where for any $i \in \mathbb{N}_n$, $ \tilde a_i =- \lambda m_x <  \mathbf 0_{d_x}< \lambda m_x= \tilde b_i$ and $\underline{a}=-(1-\lambda)m_x < \boldsymbol{0}_{d_x} < (1-\lambda)m_x=\underline{b}$.
In the second step, and from~\eqref{eq:proof_2}, we have:
\begin{equation}\label{eq:proof_33}
\tilde a_i + |\frac{\alpha_i}{\gamma_i}| \underline a \leq \Delta y^i_t + \frac{\alpha_i}{\gamma_i} \bar y^\alpha_t   \leq \tilde b_i+ |\frac{\alpha_i}{\gamma_i}| \underline b,
\end{equation}
where, according to Theorem~\ref{thm:3},  $\underline a=-\underline b$. The left-hand side of inequality~\eqref{eq:proof_33} is lower-bounded as follows:
\begin{equation}
a \leq -m_x= -\lambda m_x+ -(1-\lambda) m_x \leq  \tilde a_i + |\frac{\alpha_i}{\gamma_i}| \underline a,
\end{equation}
and its right-hand side  is upper bounded as: 
\begin{equation}
b \geq m_x= \lambda m_x+ (1-\lambda) m_x \geq  \tilde b_i + |\frac{\alpha_i}{\gamma_i}| \underline b.
\end{equation}
Thus, one has $a \leq \Delta y^i_t + \frac{\alpha_i}{\gamma_i} \bar y^\alpha_t   \leq b$. In addition,
\begin{align}\label{eq:proof_22}
& -|\alpha_i| \lambda m_x  \leq \alpha_i \Delta y^i_t \leq |\alpha_i| \lambda m_x,\\
& -|\frac{\alpha_i}{\gamma_i}|(1-\lambda) m_x\leq \frac{\alpha_i}{\gamma_i} \bar y^\alpha_t \leq  |\frac{\alpha_i}{\gamma_i}|(1-\lambda)m_x .
\end{align}
Therefore,
\begin{align}
\bar a &\leq -m_x= -|\alpha_i| \lambda m_x - \frac{\alpha_i}{\gamma_i}(1-\lambda)m_x \\
&\leq \frac{1}{n} \sum_{i=1}^n \alpha_i y^i_t=\frac{1}{n} \sum_{i=1}^n \alpha_i \Delta y^i_t+ \frac{1}{n} \sum_{i=1}^n \frac{\alpha_i^2}{\gamma_i} \bar y^\alpha_t \\
& \leq  |\alpha_i| \lambda m_x+ \frac{\alpha_i}{\gamma_i}(1-\lambda)m_x \leq m_x \leq \bar b.
\end{align}
The rest of the proof is similar to that of Theorem~\ref{thm:2}.
\end{proof}

\begin{remark}
\emph{When the optimal solution in Theorem 1 lies in
the feasible set of the proposed distributed RHC, the RHC
solution for a sufficiently large prediction horizon $H=T$
can be explicitly obtained by Riccati equations~\eqref{eq:Riccati}.
}
\end{remark} 

\subsection{Distributed and decentralized implementations}

The obtained  LQR and RHC solutions can be implemented in a distributed manner, where each agent solves two low-dimensional Riccati equations and  quadratic programmings, respectively,  and compute its action based on local (private) information $\{x^i_t, r^i_t, \alpha_i, \gamma_i\}$ and global (public) information $\{\bar x^\alpha_t, \bar r^\alpha_t, s_t, \mu\}$.

 \begin{table*}[t!]
 \begin{center}
  \caption{Differences between consensus and optimal control}\label{table:dif}
 \scalebox{1}{
 \begin{tabular}{|c|c|c|}
 \hline 
 & Consensus and distributed averaging &  Optimal control \\ 
 \hline 
 Objective & Agents reach identical value after a sufficiently  & Agents allocate resources efficiently during horizon $T$, \\ 
                  & large horizon i.e. $x^i_\infty=x^j_\infty=\bar x^\alpha_\infty, \forall i,j$. &  (not necessarily large $T$), where one may have $x^i_T \neq x^j_T \neq \bar x^\alpha_T, \forall i,j$.   \\ 
 \hline 
 Model & They often have simpler dynamics  & They often have more complicated dynamics  \\ 
 & and time-invariant cost functions. & with hard constraints and time-varying cost functions.\\
 \hline 
Information & Dissemination of information is via  many local &  Dissemination of information is via one-shot cloud-based server\\ 
& interactions (not suitable for costly communications)& (not suitable for hard constrained communication graph)\\ 
 \hline 
 Solution approach & Infinite product of stochastic matrices & Dynamic program and receding horizon control \\ 
  & (often doubly stochastic matrices)  & (often Riccati equation and quadratic programming) \\ 
 \hline 
 \end{tabular} }
   \end{center}
  \end{table*}
\subsubsection{Stochastic model \& certainty equivalence} Suppose that  the dynamics~\eqref{eq:dynamics} have additive  noises such that
$x^i_{t+1}= A_t x^i_t+B_t u^i_t+w^i_t, \forall i \in \mathbb{N}_n$,
where $\{w^i_t\}_{1:T}$ is an independent stochastic process. This generalization does not affect the solution of Problem~1 because of the certainty equivalence theorem. In a such case, there is no loss of optimality in replacing the noises with their  expectations. For Problem 2, however, certainty equivalence theorem does not hold. Nonetheless, one can  use  the certainty equivalence approximation (where the noise is replaced by  its expectation) to convert the stochastic  dynamics to deterministic ones and establish recursive feasibility~\cite{bernardini2011stabilizing,rawlings2017model,mesbah2019stochastic}. When it comes to   distributed implementation of the stochastic model,  each agent at every time $t$ only requires to observe  the deep state $\bar x^\alpha_t \in \mathbb{R}^{d_x}$  (whose size is independent of the number of agents unlike the centralized joint state $(x^1_t,\ldots,x^n_t)$).

\subsubsection{Two-time-scale distributed (consensus-based) solution} Suppose that the agents' communication graph  does not allow the immediate observation of the deep state. In this case, one can use a two-time scale distributed optimization strategy. At each time instant $t$, agents  run a consensus algorithm to  compute the deep state after a sufficiently large number of iterations. Such a two-time-scale distributed implementation is practical in many control applications, especially where the communication (information) process is significantly  faster than the physical (control) process.
 
\subsubsection{Fully decentralized information in asymptotic model}
It is not always  feasible to share the deep state among agents specially when the number of agents is very large.
In such a case,  the deep state  can be predicted (rather than communicated)  by the infinite-population approximation because the dynamics of the  infinite-population model is  deterministic due to the strong law of large numbers. See for example~\cite{Jalal2021CDC_KF,JalalCDC2018},  where the predicted case (sub-optimal solution) converges to the communicated case (optimal solution) at the rate~$1/n$. This leads to a fully decentralized control structure, where each agent needs to observe  only its local  information. 

\section{Multiple sub-populations and  features}\label{sec:multiple}

So far, we have assumed that the matrices in  dynamics~\eqref{eq:dynamics} and cost function~\eqref{eq:per_step_cost} are identical for all agents (i.e., one population) and the agents are coupled through one set of factors (i.e., one feature). In this section, we briefly discuss the generalization of our main results  to cases with multiple sub-populations and multiple features.
\subsection{Multiple sub-populations}

Consider a population consisting of $S \in \mathbb{N}$ sub-populations, where  matrices in dynamics and cost functions of agents in  each sub-population are identical.  In such a case, the unconstrained optimization problem gets decomposed into $S+1$ smaller LQR problems~\cite{Jalal2019risk}. Similarly, one can propose $S+1$ distributed RHC problems. See~\cite{JalalCCECE2018,JalalCDC2018} for an example with two sub-populations where one sub-population contains one leader and another sub-population  a large number of followers.

\subsection{Multiple features}
Consider a population where agents are coupled through $f \in \mathbb{N}$ sets of influence factors (features). For example, any directed weighted graph can be decomposed by singular value decomposition and any undirected weighted graph by spectral decomposition, respectively, where features represent the singular vectors and eigenvectors. In such a case, every feature may be viewed as a \emph{virtual} sub-population; hence,  the unconstrained optimization problem decomposes into $f+1$ LQR problems~\cite{Jalal2019risk}. Analogously, one can construct $f+1$ distributed RHC problems. In what follows, we present two problems with more than one set of features.

\subsubsection{Cohesive cost function}\label{sec:cohesive}
It is possible to add a cost function to~\eqref{eq:per_step_cost} in order to incorporate the cohesiveness of the swarm where the team cost becomes $\beta \bar c_t  + (1-\beta)  c^L_t$, $\beta \in [0,1]$ and  the cohesive cost  $ c^L_t$ is a quadratic function of the relative distances, i.e.
 \begin{equation}\label{eq:general_cost_cohesive}
  c^L_t:=  \mathbf x_t^\intercal L \mathbf x_t= \frac{1}{2}\sum_{i=1}^n \sum_{k=1}^n (x^i_t-x^k_t)^\intercal L^{i,k} (x^i_t -x^k_t),
 \end{equation}
 where $L$ is not  necessarily a symmetric matrix. However, a special case of the cohesive cost function is Laplacian, i.e.
 \begin{align}
c^L_t= \mathbf x_t^\intercal L \mathbf x_t&=: \frac{1}{2}\sum_{i=1}^n \sum_{k=1}^n  A(i,k)(x^i_t-x^k_t)^\intercal (x^i_t -x^k_t)\\
&=\sum_{(i,k) 
\in \mathcal{E}} (x^i_t-x^k_t)^\intercal (x^i_t -x^k_t), 
 \end{align}
where $A$ and  $L$ are   the adjacency and  Laplacian (symmetric) matrices of  an undirected weighted graph and $\mathcal{E}$ is  the  edge set. 
The same analogy holds for a more general cost function
\begin{equation}
c^L_t=:\mathbf x_t^\intercal \DIAG(\alpha_1,\ldots,\alpha_n)^{-1}L \mathbf x_t,
\end{equation}
resulting in  consensus to a weighted average~\cite{Olfati2007survey,tsitsiklis1986distributed,bertsekas1989parallel,Jadbabaie2003}.   In addition, one can use  decomposition methods such as
\begin{equation}
\hspace{0cm} \mathbf x_t^\intercal L \mathbf x_t=:\sum_{j=1}^n \sigma_j  \mathbf x_t^\intercal U_j V_j^\intercal \mathbf x_t, \hspace{.1cm} \text{(singular value decomposition)}
\end{equation} 
 \begin{equation}
\hspace{0cm} \mathbf x_t^\intercal L \mathbf x_t=:\sum_{j=1}^n \lambda_j^2 \mathbf x_t^\intercal V_j V_j^\intercal \mathbf x_t,\tag{spectral decomposition}
 \end{equation}
 to decompose~\eqref{eq:general_cost_cohesive} and restrict attention to a few dominant  features associated with the largest singular values and eigenvalues, respectively.  Although consensus and optimal control  are two different problems (see~Table~\ref{table:dif} for a few differences),  they are related in some sense. In particular,  the consensus problem may be formulated as a linear time-invariant system with integrators and an infinite-horizon time-average quadratic cost function.   In such a case, the consensus strategy makes all the relative distances  as well as any tracking distance from  the consensus value  go to zero~\cite{Cao2010}. 
On the other hand, the optimal control strategy  may be viewed as a solution to the problem of finding the best  topology  for the communication graph with  quadratic similarity index; see~\cite[Corollary 1]{JalalCCECE2018}, for example.

\subsubsection{Soft structural constraint}
 It is possible to add a soft-constraint term  to~\eqref{eq:per_step_cost}  in order  to take into  account the structure of the control strategy, where the hard constraint $\mathbf u_t= H \mathbf x_t$ is replaced by the quadratic soft constraint $(\mathbf u_t -H \mathbf x_t)^\intercal (\mathbf u_t -H \mathbf x_t)$.  Analogously, one can consider the structure of dynamics, where  $\mathbf x_{t+1}= S \mathbf x_t$ is replaced by  $(\mathbf x_{t+1} -S \mathbf x_t)^\intercal (\mathbf x_{t+1} -S \mathbf x_t)$. Therefore,  one can use the singular value decomposition and spectral theorem to  generate dominant features associated with the above quadratic  cost functions, similar to those in the cohesive cost function.

\section{Cyber-Physical attacks}\label{sec:cyber}
 In this section, we propose a new  class of  cyber-physical attacks  formulated as perturbed influence factors. Let $z_i \in \mathbb{R}$ be the status of the attack associated  with agent $i  \in \mathbb{N}_n$.  Let $\tilde \alpha_i:=\alpha_i(z_i)$  denote  the  \emph{attack factor}, which is a function of~$z_i$.   Depending on  the  attack function, we can   define different types of attacks. Below, we mention a few cases.
\begin{itemize}
\item \textbf{Denial of service}. This is when  $z_i=0$, if agent $i$  is  attacked, and  $z_i=1$, if  not attacked, where $\tilde \alpha_i:=\alpha_i z_i$.  
\item \textbf{Leader attack}. This is when one (leading) agent is targeted, i.e., the one with the largest influence factor. 
\end{itemize} 
 We can also  define various defence mechanisms as follows.
\begin{itemize}
\item \textbf{Isolated mechanism}. In this case,  agent $i$ is dispensable; hence,  it gets isolated by choosing a relatively small  value (i.e., close to zero) for factor $\tilde \alpha_i$.  Subsequently,  agent~$i$ has a negligible effect in the center of the swarm and will be ignored  by  the swarm. 
\item \textbf{Protected mechanism}. In this scenario,  agent $i$ is important; hence,  it gets protected  by other agents via making its factor $\tilde \alpha_i$ considerably larger. As a result, agent~$i$ would have a significant effect in the center of the swarm.  In particular, the larger $\tilde \alpha_i$, the closer the center of the swarm is to agent $i$. This mechanism is helpful for situations in which  other agents must cover the attacked agent by moving  to its vicinity.
\end{itemize} 
\begin{remark}
\emph{In practice, one can extend the above setup to  time-varying attacks by developing  a two-time-scale framework, where  attacks  are occurred  in the slower scale and the RHC (or MPC) is deployed in the faster scale.}
\end{remark}

\section{Simulations}\label{sec:numerical}

 \begin{figure}[t!]
\hspace{-.6cm}
\vspace{-1cm}
\scalebox{.55}{
\includegraphics[trim=0cm 0cm 0cm 1cm, clip,width=\textwidth]{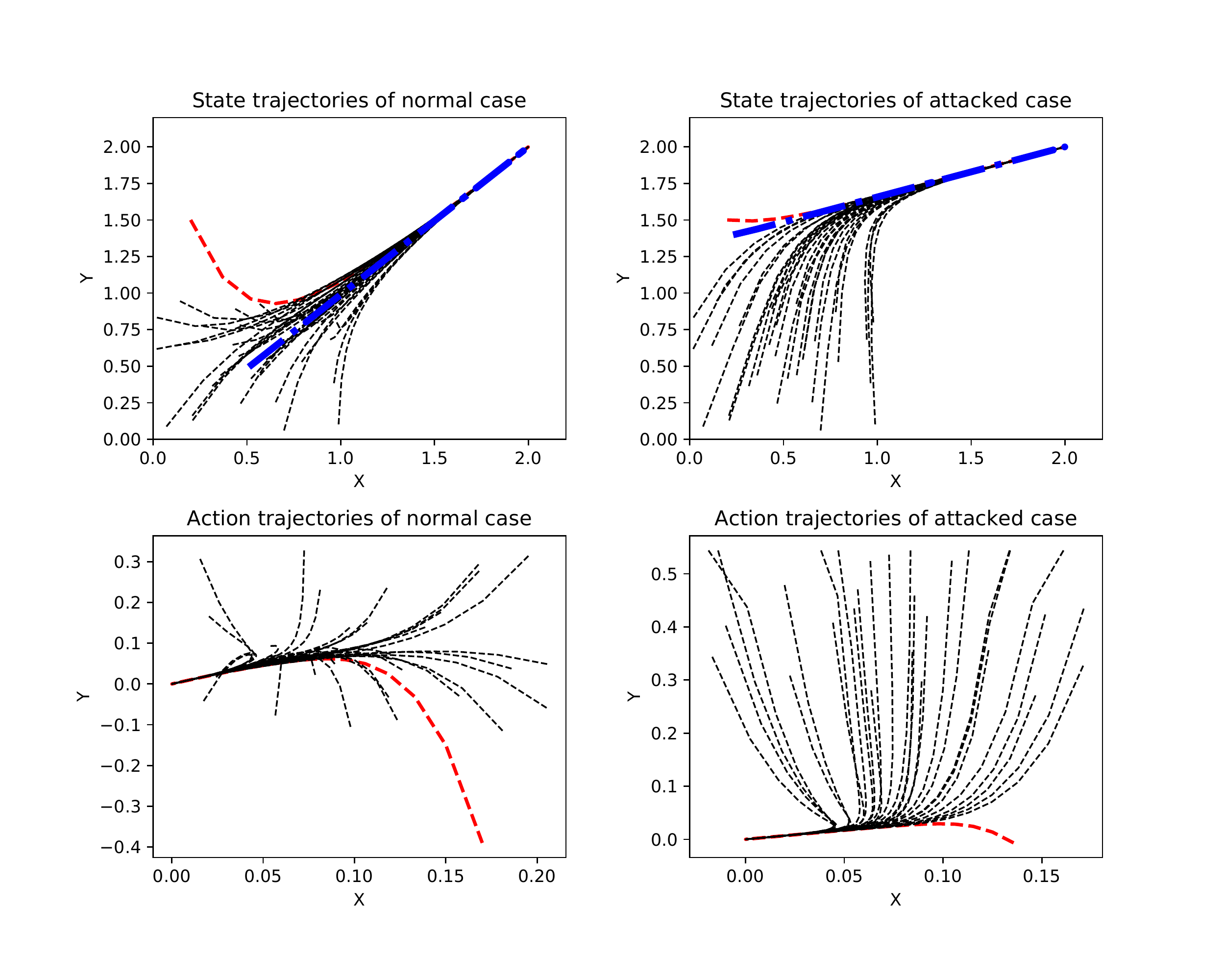}}
\vspace{-.3cm}
\caption{Unconstrained case in Example 1.}\label{fig:unconstrained}
\vspace{-.4cm}
\end{figure}
\begin{figure}[t!]
\hspace{-.6cm}
\scalebox{.55}{
\includegraphics[trim=0cm 0cm 0cm 0cm, clip,width=\textwidth]{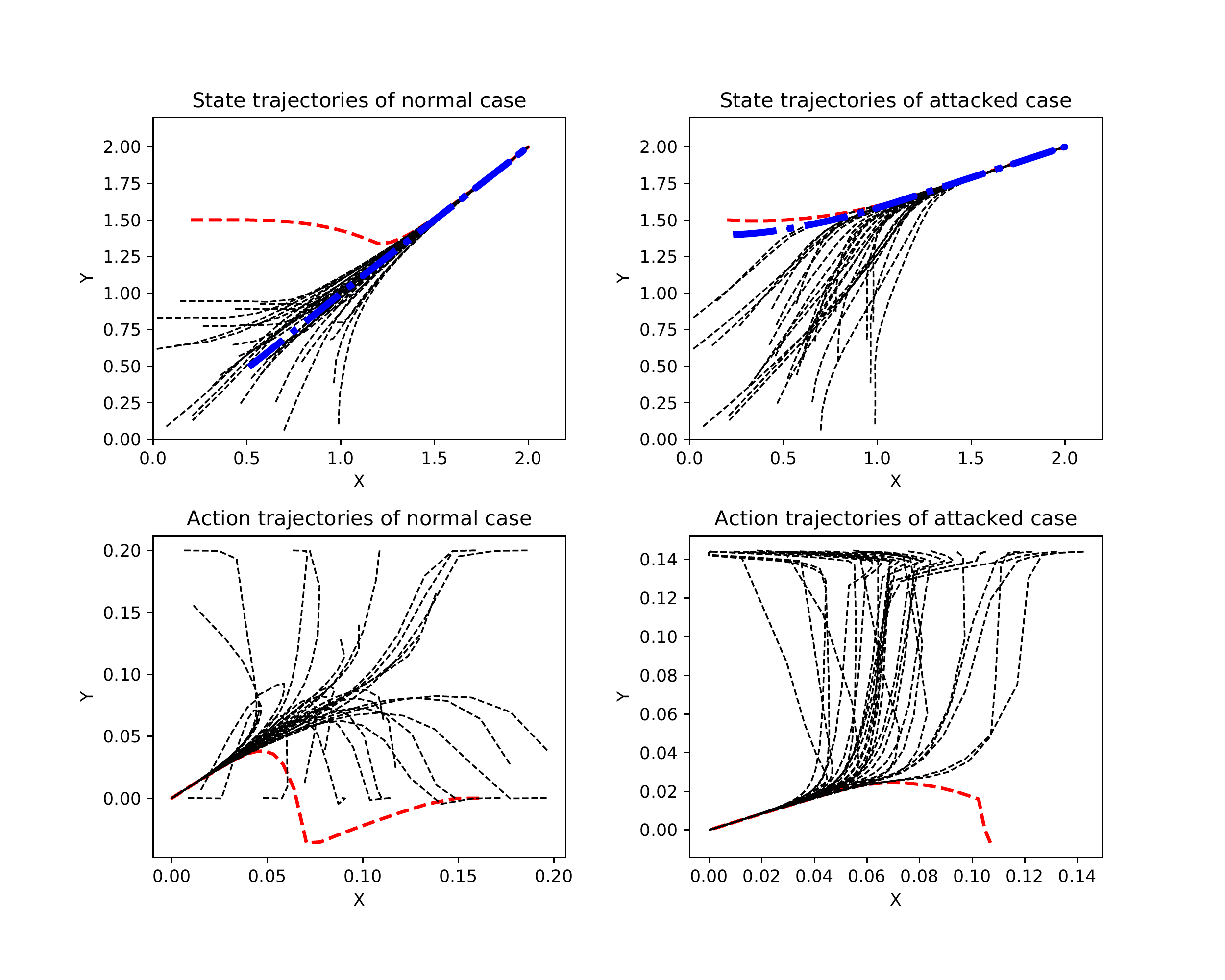}}
\vspace{-1cm}
\caption{Constrained case in Example 1,  where   $|u^i_t|<0.2$, $i \in \mathbb{N}_n$. }\label{fig:constrained}
\vspace{-.2cm}
\end{figure}
\textbf{Example 1}. Consider a group of  robots that are interested to move towards a target collectively.  Let the influence factors $\alpha_i \geq 0$, $ i \in \mathbb{N}_n$,   construct a center of mass  i.e.  $\bar \alpha=1$; hence, tracking is strong according to Definition~\ref{def:strong}. In our simulations,  control horizon is  $T=100$ and number of robots is $n=100$.   Let  the dynamics of the robots be linearised such that  $A=B=\DIAG(1,1)$, and their team cost function be defined as follows:
\begin{equation}
\frac{1}{n} \sum_{t=1}^T  \sum_{i=1}^n  \alpha_i(\|x^i_t - \bar x^\alpha_t \|_{Q} + \|u^i_t\|_{R}) + \|\bar x^\alpha_t -s\|_{\bar Q},
\end{equation}
where $Q=\DIAG(5,50)$, $\bar Q=\DIAG(1,1)$ and  $R=\DIAG(100,100)$.  In addition, we consider a case  in which one robot  is physically attacked  and other $(n-1)$ robots follow a protected mechanism to cover it, as described in Section~\ref{sec:cyber}. 
Let $z_i = 1$ denote that  agent $i$ is attacked and  $z_i = 0$ denote that  it  is not, $i \in \mathbb{N}_n$.   In this case, the perturbed  influence factor of  robot $i \in \mathbb{N}_n$ can be  defined as:
\begin{equation}
\tilde \alpha_i:= n \rho z_i  +\frac{n}{n-1} (1-\rho) (1- z_i),
\end{equation}
where $\frac{1}{n}\sum_{i=1}^n \tilde{\alpha}_i=1$ and  $\rho \in [0,1]$   determines  the level of protection. The larger~$\rho$,   the closer  the center of mass is  to the attacked (targeted) robot, providing more protection.

  The results of  our simulations are depicted  in Figures~\ref{fig:unconstrained} and~\ref{fig:constrained},   where we display only  $30$ out of $100$ robots to ease the exposition. In these figures, the blue dotted line is the trajectory of the center of mass and the red dashed line is that of the attacked robot.   In particular,  it is shown in  Figure~\ref{fig:unconstrained} that the robots can collectively  reach the target $s=(2,2)$  in the normal case (where influence factors are homogeneous $\alpha_i=1$) as well as  the attacked case (where the perturbed  influence factors are calculated for  $\rho=0.9$). Furthermore,  we consider a similar setting wherein control signals are bounded such that $|u^i_t|<0.2$, $\forall i$. To solve the resultant problem, we use  quadratic programming to find a solution for the proposed  local and global RHCs, where $H=10$ and $\lambda=0.5$.
It is demonstrated in Figure~\ref{fig:constrained} that  the robots can  collaboratively  reach the target  while respecting their control constraints.

\section{Conclusions and future  directions}\label{sec:conclusions}
We introduced deep structured tracking for a large number of decision-makers, where the interaction between them is modelled by influence factors. The  influence factors can   represent  physical features and constraints  (e.g.,  the center of swarm) as well as non-physical ones  (e.g., adhesive behaviour of the swarm). 
 For the unconstrained and constrained cases, two low-dimensional  solutions were proposed.    On the one hand, the unconstrained solution  was   shown to be optimal,  obtained by solving two scale-free Riccati equations,  where  its extension to the infinite-horizon cost function is straightforward. On the other hand,  the constrained solution  took affine constraints into account, where establishing  its stability is  difficult  due to the time-varying nature of the solution.  In addition, the main results were generalized  to multiple sub-populations and multiple features.

There are several possible future directions. For example, one can consider (a)  different forms of dynamics (e.g. aerial and ground  vehicles) and cyber-physical attacks (e.g. time-varying attacks with two-time-scale framework, denial of service,  or minmax optimization with adversarial player); (b)   output feedback,   H$_2$  and H$_\infty$  control algorithms; (c)  a more general model  with  non-symmetric weighting matrices as well as non-quadratic and non-convex cost  functions using  interior-point methods; (d)  constrained reinforcement learning and data-driven approaches, and (e) the investigation of  the optimal feasible set (e.g., best $\lambda$ in Theorems~\ref{thm:2} and~\ref{thm:3}).

\bibliographystyle{IEEEtran}
\bibliography{Jalal_Ref}
\end{document}